[1994/12/01]
\documentclass[draft]{article}
\title{Spectral theory of $p-adic$ Hermite operator}
\author{Zhao tianhong}
\date{2022/10/20}

\usepackage{amsmath}
\usepackage{amsfonts}
\usepackage{amsmath}
\usepackage{amssymb}
\usepackage{mathtools}
\usepackage{amsthm}
\usepackage{pgfplots}
\usepackage{enumitem}
\usepackage{tikz}   
\usepackage{tikz-cd}  
\usepackage{changepage}

\chardef\bslash=`\\ 

\newtheorem{thm}{Theorem}[section]

\newtheorem{lem}[thm]{Lemma}
\newtheorem{prop}[thm]{Proposition}

\theoremstyle{definition}
\newtheorem{defn}{Definition}[section]

\theoremstyle{remark}
\newtheorem{rmk}{Remark}[section]
\newtheorem*{notation}{Notation}

\newcommand{\CC}{\mathbb{C}}
\newcommand{\Z}{\mathbb{Z}}
\newcommand{\Q}{\mathbb{Q}}
\newcommand{\R}{\mathbb{R}}

\newcommand{\eval}[2][\right]{\relax
  \ifx#1\right\relax \left.\fi#2#1\rvert}

\begin{document}
\maketitle
\markboth{Spectral theory of $p-adic$ Hermite operator}
{Spectral theory of $p-adic$ Hermite operator}
\renewcommand{\sectionmark}[1]{}
\begin{abstract}
We give the definition of $p-adic$ Hermite operator and set up the $p-adic$ spectral measure. We compare the Archimedean case with non-Archimedean case. The structure of Hermite conjugate in $C^{*}$-Algebra corresponds to three canonical structures of $p-adic$ ultrametric Banach algebra: 1. mod $p$ reduction \ 2. Frobenius map \ 3. Teichmüller lift. There is a nature connection between Galois theory and Hermite operator spectral decomposition. The Galois group $\mathrm{Gal}(\bar{\mathbb{F}}_p|\mathbb{F}_p)$ generate the $p-adic$ spectral measure. We point out some relationships with $p-adic$ quantum mechanics: 1. creation operator and annihilation operator \ 2. $p-adic$ uncertainty principle.
\end{abstract}
\section{Introduction}
The study of Hermite operator is in the center of quantum mechanics. The usual Hermite operator corresponds to the observable theory in physics. In mathematics we have the spectral decomposition theorem of usual Hermite operator. Is it possible to build up a $p-adic$ quantum  theory by $p-adic$ Hermite operator?
\\
\\
The positive property of a square number allows us to define positive definite quadrics in $\R$-linear space. As a result, we can define Hilbert spaces over $\R$. If we consider the Hilbert space $\mathcal{H}$ over $\CC$, we need complex conjugate in Galois group $\mathrm{Gal}(\mathbb{C}|\mathbb{R})$ to ensure the positive definiteness of inner product. The positive definiteness of inner product allows a probability annotation of quantum mechanics. The Riesz representation theorem induces the Hermite conjugate $\dagger$ structure of bounded operator over $\mathcal{H}$.
\\
\\
When it comes to $p-adic$ fields(e.g.$\Q_p,\CC_p$), the same argument does not hold. The reason is that quadratic form over $\Q_p$ always have a non-trivial zero when the rank of quadratic form is large enough($rank \ge 5)$(\cite{Serre},p36,Theorem 6). So if we have a quadratic form over a $\Q_p$-Linear space, in general it cannot induce a norm. There are some papers to consider ultrametric Hilbert space but we will not use that definition\cite{UHS}.
\\
\\
The reasonable analogy of Hermite conjugate is the $p-adic$ continuity of Frobenius map $\sigma$\cite{small,ZpDiff,pDG}. \textbf{Fermat's little theorem} tells us:
$$\forall x \in \Z,\sigma(x)=x^p\equiv x \pmod{p} $$
The same statement holds for $\Z_p$, since $\Z$ is dense in $\Z_p$, we have:
$$\forall x \in \Z_p,\sigma(x)=x^p\equiv x \pmod{p} $$
In general, mod $p$ reduction is a ring homomorphism:$\phi:\Z_p \longrightarrow  \mathbb{F}_p $. There exists a canonical section named \textbf{Teichmüller lift}:$i:\mathbb{F}_p \longrightarrow \Z_p$ such that:$\phi \circ i=id_{\mathbb{F}_p }$. The section $i$ is multiplicative. We say $x \in \Z_p$ is a Teichmüller element if $x$ is a lift element of $\mathbb{F}_p$. Let $T(\Z_p)$ be the set of Teichmüller element. Let $\sigma$ be the Frobenius map:$\sigma:\Z_p \to \Z_p,x\mapsto x^p$. We can show that:
\begin{prop}
$$T(\Z_p)=\left \{  x\in \Z_p,\sigma(x)=x \right \},Card(T(\Z_p))=p$$
\end{prop}
We refer to \cite{Witt,connes} for theory of Teichmüller lift and Witt ring.\\ \\ 
Recall that the canonical filtration $\Z \supset p\Z \supset p^2\Z...$ induces a $p-adic$ topology of $\Z$ which is Hausdorff. $\Z_p$ is the completion of $\Z$ under this topology.  Equivalently, we can define a $p-adic$ norm:
$$\left | x \right |_p= p^{-v_p(x)},v_p(x)=\sup \left \{ k, p^k\mid x \right \} $$

\begin{thm}
Let $A$ be a commutative integral ring with unit,$\mathrm{Char}A$=0,equipped with $p-adic$ topology(induce by the filtration $A \supset pA \supset p^2A...$ ),suppose $A$ is complete.Let $\tilde{A}=A/pA$.The following conditions are equivalent:
\begin{enumerate}
\item $$\forall x \in A,\sigma(x)=x^p\equiv x \pmod{p}$$
\item $$\forall \overline{x}  \in \tilde{A},\sigma(\overline{x})=\overline{x}^p=\overline{x} $$
\item $$\forall x \in A,x=\sum_{i=0}^{\infty} x_ip^i,\sigma(x_i)=x_i$$
\end{enumerate}
\end{thm}
\begin{proof}
We will only show that $1\Rightarrow 3$. Suppose $ x,y \in A$ satisfy:$\left | x -y\right |_p \le \frac{1}{p}$. We have:
$$\left | \sigma(x)-\sigma(y)\right |\le \frac{\left | x-y\right |}{p} $$
We can check $\left \{  \sigma^k(x)\right \} _{k=0,1,2...}$ is a Cauchy sequence, hence it has a limit. We call it $x_0$. It's easy to show:
$$x \equiv x_0 \pmod{p}$$
So we have:
$$x= x_0+py,y \in A$$
Repeating the same argument to $y$, we have done.
\end{proof} 
\begin{rmk}
$\Z$ satisfy the property:$\forall x \in \Z,\sigma(x)=x^p\equiv x \pmod{p}$ but $\Z$ is not  complete in $p-adic$ topology. $\Z_p, \Z_p^n$(Equipped with supremum norm) satisfy this property above. Let $C(\Z_p,\Z_p)$ be the set of continous function from $\Z_p$ to $\Z_p$, equipped with supremum norm. $C(\Z_p,\Z_p)$ has this property.
\end{rmk}
\begin{defn}
Let $A$ be a ring(with a unit,not necessary commutative), $\mathrm{Char}A$=0, equipped with $p-adic$ topology which is given by filtration:$A \supset pA \supset p^2A...$. Suppose $A$ is complete. Let $\sigma:A\to A,x\mapsto x^p$ be the Frobenius map. For any positive integer $N$, we define:
$$T_N (A)=\left \{x\in A\mid  \sigma^N(x)=x \right \} $$
$$T(A)=\bigcup_{N=1}^{\infty} T_{N!} (A)$$
We call set $T_N (A)$ Teichmüller element with period $N$.
\end{defn}
\begin{rmk}
We have:
$$T_N (A)\subseteq T_{N^*} (A)\Longleftrightarrow N\mid N^*$$
\end{rmk}
Let $K$ be the completion of maximal unramified extension of $\Q_p$ with a multiplicative norm:$|\ \ |:K \to \R_+$. Such norm exists and is unique. Let $X$ be an ultrametric Banach space over $K$ such that the range of norm $|\ \ |:X \to \R_+$ on $X$ satisfy: $|X|=|K|$ \\ \\
The following result (\cite{Berk},Corollary 9.2.7) is useful for our discussion.
\begin{thm}(Berkovich).Let A be a uniform commutative Banach algebra over $K$ with identity,  such that all the characters of A take values in $K$. Then the space M(A) is totally disconnected, and the Gelfand transform gives an isomorphism $A \overset{\sim }{\to}C(M(A),K)$
\end{thm}
We refer to SGA for the inspired result of Galois descent. The following is a generalization of Galois theory. Let $k$ be a field. Let $\mathfrak{G}_k$ be the profinite group $\textup{Gal}(k_s|k)$.
\begin{thm}(Grothendieck).The functors
$$
\begin{tikzcd}
\left \{\textup{finite } \mathfrak{G}_k-sets \right \}^{opp} \arrow[r] &
\left \{ \textup{\'etale $k$-algebras} \right \} 
\end{tikzcd}
$$
$$
\begin{tikzcd}
\ \ \ \ \ \ \ \ \ \ \ \ \ \ \ \ \ \ \ \ \ \ \ \ \ \ \ \ \ \ \ \ \ S \arrow[mapsto]{r} 
& \textup{Hom}_{\mathfrak{G}_k-sets}(S,k_s)=\textup{Hom}_{sets}(S,k_s)^{\mathfrak{G}_k}
\end{tikzcd}
$$
$$
\begin{tikzcd}
\textup{Hom}_{k-\textup{algebras}}(L,k_s) & L \arrow[mapsto]{l}
\end{tikzcd}
$$
are inverse equivalences of categories.
\end{thm}
\begin{defn}
Let $x$ be a linear operator on $X$,We call $x$ is a Teichmüller element of period $\infty$, if:
$$s-\lim_{N\to \infty} \sigma^{N!}(x)= x,|x|=1$$
Let $T_{\infty}(X)$ be the set of Teichmüller element of period $\infty$.
\end{defn}
\begin{rmk}
Recall the Galois theory of finite fields:
$$\mathrm{Gal}(\bar{\mathbb{F}}_p|\mathbb{F}_p)\simeq \lim_{\underset{N\ge 1}\longleftarrow}\mathbb{Z}/N!\mathbb{Z}$$
Teichmüller element will be used to build up orthogonal spectral measure. 
\end{rmk}
Let $K$ be an unramified extension of $\Q_p$ such that $\omega^{p^N}=\omega$ has exactly $p^N$ solutions. Let $A$ be a commutative $K$-Banach algebra with unit, $\sigma$ be the Frobenius map:$\sigma:A \to A ,x\mapsto x^p$. We assume the norm on $A$ satisfy:$\left | A\right | =\left |K\right | $. Let $A_0$ be the unit ball: $A_0=\left \{ x\in A, \left | x \right |\le 1 \right \} $. $A_0$ is a natural $\mathcal{O}_K$-algebra. Moreover, $A_0$ is complete with $p-adic$ topology. Let $a$ be an arbitrary element in $K$. By using the theory of Witt ring, we have a taylor expansion:
$$
a=\sum_{i=k}^{\infty} a_ip^i,a_i \in T(\mathcal{O}_K)
$$
\begin{defn}
Let $x \in A$, We say $x$ is a \textbf{$p-adic$ hermite operator} with period $N$, if $x$ has a taylor expansion:
$$x=\sum_{i=k}^{\infty} x_ip^i,x_i \in T_N(A_0)$$
\end{defn}
Such $x$ has a $p-adic$ \textbf{orthogonal spectral measure}. This is our main theorem:
\begin{thm}($p-adic$ Hermite operator spectral decomposition theorem)\\
Let $x$ be a $p-adic$ Hermite operator with period 1, then there exists a finitely additive, orthogonal projection valued spectral measure over $\Q_p$. There exists a spectral integral:
$$
I=\int_{\Q_p}dE_{\lambda }  \ \ \   x=\int_{\Q_p}\lambda dE_{\lambda }
$$
\\
Let $x$ be a $p-adic$ Hermite operator with period N, then there exists a finitely additive, orthogonal projection valued spectral measure over $K$. There exists a spectral integral:
$$
I=\int_{K}dE_{\lambda }  \ \ \   x=\int_{K}\lambda dE_{\lambda }
$$
\end{thm}
\begin{thm}(Teichmüller element spectral decomposition theorem)\\
Let $x \in A_0$ be a element. $x$ is a Teichmüller element of period $N$ if and only if there exists a spectral decomposition:
$$
\sum_{\lambda\in T_N(K)}\pi_\lambda=1,\ \ \sum_{\lambda\in T_N(K)}\lambda\pi_\lambda=x,\ \ \pi_\lambda^2=\pi_\lambda,\ \ \pi_\lambda\pi_{\lambda^*}=0(\forall \lambda \ne \lambda^*)
$$
$K$ is an unramified extension of $\Q_p$ such that $\omega^{p^N}=\omega$ has exactly $p^N$ solutions. $\pi_\lambda$ is non-Archimedean orthogonal projection.
\\
\\
Let $X$ be an ultrametric Banach space over $K$, $K=\widehat{\mathbb{Q}_p^{unram}}$,$\widehat{\mathbb{Q}_p^{unram}}$ is the completion of maximal unramified extension of $\Q_p$.
Let $x$ be a linear operator on $X$. $x$ is a Teichmüller element of period $\infty$ if and only if there exists a spectral decomposition:
$$
\sum_{\lambda\in T(K)}\pi_\lambda=1,\ \ \sum_{\lambda\in T(K)}\lambda\pi_\lambda=x,\ \ \pi_\lambda^2=\pi_\lambda,\ \ \pi_\lambda\pi_{\lambda^*}=0(\forall \lambda \ne \lambda^*)
$$
$T(K)=\bigcup_{N=1}^{\infty} T_{N!}(K)$ is the union of Teichmüller element in $K$.$K=\widehat{\mathbb{Q}_p^{unram}}$,$\widehat{\mathbb{Q}_p^{unram}}$ is the completion of maximal unramified extension of $\Q_p$.\\
The sum converges in strong operator topology.
\end{thm}
Let $\mathcal O_K=\left \{x\in K,|x| \le 1 \right \} $. Let $M_n(\mathcal O_K)$ be the set of $n \times n$ matrix of $\mathcal O_K$ coefficients. We can define a norm: $A=(a_{ij})_{1\le i,j \le n},|A|=\underset{1 \le i,j\le n}{\sup }|a_{ij}|$.
\begin{thm}
Let $A \in M_n(\mathcal O_K)$ be a matrix. Suppose $|A|=1$. There exists a canonical Jordan decomposition:
$$
A=A_s+A_n
$$
$A_s \in T(M_n(\mathcal O_K))$ is a Teichmüller element with finite period, $A_n$ is a topological nilpotent element such that:
$$|A_n| \le 1, \lim_{k \to \infty} \sigma^{k!}(A_n)=0$$
\end{thm}
Let $A,B$ be $p-adic$ Hermite operators on $X$. Let $\psi \in X,|\psi|=1$ be the normalized wave function. Let $diam(A),diam(B)$ be the spectrum diameter of $A,B$: 
$$
diam(A) =\underset{\lambda ,\lambda^{*}\in SpecA}{\sup}|\lambda-\lambda^{*}|
$$
\begin{thm}($p-adic$ uncertainty principle)
We have:
$$
\left | [A,B]\psi \right |  \le diam(A)*diam(B)
$$
\end{thm}
A.N.Kochubei developed a specral theory for non-Archimedean normal operator\cite{NANO,NAOA}. We consider three canonical structures: mod $p$ reduction,\ Frobenius map, \ Teichmüller lift. 
There are some papers about $p-adic$ quantum mechanics\cite{pQM,pMP}. It is conjectured that the space-time is non-Archimedean, or even $p-adic$ in Planck length\cite{pQM}. The $p-adic$ AdS-CFT is an important work\cite{ADSCFT}. We refer to \cite{Acourse} for the basic concepts of $p-adic$ analysis, \cite{Berk,AE} for the theory of non-Archimedean Banach algebra.\\ \\  
In section 2 we will build up the spectral measure of Teichmüller element.\\ \\
In section 3 we will set up the spectral decomposition theorem of $p-adic$ Hermite operator.\\ \\
In section 4 we will give some examples of $p-adic$ Hermite operator. We will discuss the non-Archimedean creation and annihilation operators and give a simplified proof of some examples given by A.N.Kochubei \cite{NANO}. \\  \\ 
In section 5 we will give a proof for the $p-adic$ uncertainty principle for $p-adic$ Hermite operator.\\  \\
In section 6 we will give some further discussions.\\ \\

Here is a table comparing the Archimedean case with the non-Archimedean case. \\ 
\begin{table}[!htbp]
\centering
\begin{tabular}{|c|c|c|} 
\hline
   & Archimedean & Non-Archimedean \\ 
 \hline
base field & $\R,\CC$ & $\Q_p,K$,$K$ is an unramified extension of $\Q_p$\\
 \hline
Canonical Norm(base field)&  usual norm & norm induced by extension $K|\Q_p$\\ 
 \hline
Banach space & Hilbert space & ultrametric Banach space\\
 \hline
Canonical Norm & $|x+y|^2+|x-y|^2=2(|x|^2+|y|^2)$ & $|x+y|\le max(|x|,|y|)$\\ 
 \hline
Orthogonal projection & $\pi^2=\pi,|x|^2=|\pi(x)|^2+|\pi^{\perp}(x)|^2$ & $\pi^2=\pi,|x|=max(|\pi(x)|,|\pi^{\perp}(x)|)$ \\ 
 \hline
Banach algebra & $C^{*}$-Algebra & ultrametric Banach algebra\\
 \hline
Galois action & Hermite conjugate $\dagger$ & Frobenius map $\sigma$\\
 \hline
\end{tabular}
\end{table}

\section{Orthogonal projection and Teichmüller element}
Let $F$ be a local field($\R,\CC$ or $\Q_p,K$, $K$ is a finite field extension of $\Q_p$) with a non-trival multiplicative norm: $| \ \ | \to \R_{+}$, Suppose $\mathrm{Char}K$=0. We assume the norm is discrete when $K$ is a finite field extension of $\Q_p$). Let $(X,|\ \ |)$ be a $F$-Banach space such that the range of norm satisfy: $|X|=|F|\subseteq \R_{+}$. In the Archimedean case we assume $X$ is a Hilbert space. In the Non-Archimedean case we assume $X$ is a ultrametric Banach space, i.e.the norm $|\ \ |$ over $X$ satisfies:
 $$|x+y|\le max(|x|,|y|)\ \ \forall x,y \in X$$
For simplicity, we assume $F=\Q_p$. It's not hard to generalize this case to finite field extension of $\Q_p$ (because the theory of reduction and Teichmüller lift are similar). 
Let $X_0$ be the set of unit ball of $X$:$X_0=\left \{ x \in X,|x|\le 1 \right \} $
Let $A$ be the set of continous linear map over $X$: $A=\left \{ f:X \to X  \ contious,K-linear \right \} $ which is a canonical $K$-Banach algebra. Let $A_0$ be the unit ball of $A$:$A_0=\left \{ x \in A,|x|\le 1 \right \}$. When $K$ is non-Archimedean, $A_0$ is a nature $\Z_p$-Algebra which can be reduced to $\tilde{A}=A_0/pA_0$. $\tilde{A}$ has a nature $\mathbb{F}_p$-Algebra structure.
\begin{defn}
A projection $\pi$ of $A$ is a element satisfty:$\pi^2=\pi$ ,we can also define coprojection:$\pi^{\perp}=I-\pi$ . They have exact relationship:$ker\pi=im\pi^{\perp},ker\pi^{\perp}=im\pi$.
\end{defn}
\begin{thm}
For a projection $\pi$,the following are equivalent:
\begin{itemize}
\item[1.] $\left | \pi \right | =1$
\item[2.] $\pi:X_0 \to X_0,X_0=\left \{ x \in X,\left |x \right |\le 1 \right \} $
\item[3.] $X=\pi(X) \hat{\oplus} \pi^{\perp}(X)$
\item[4.] (Archimedean)$\pi^{\dagger}=\pi$
\item[4*.](non-Archimedean)$\exists \bar{\pi} \in \tilde{A},\bar{\pi}^2=\bar{\pi},\pi \equiv \bar{\pi}  \pmod{p} $
\end{itemize}
\end{thm}
Let's explain these conditions.Condition 1 says the operator norm is 1. Condition 2 says $\pi$ maps the unit ball to the unit ball.Condition 3 says $X$ has a direct sum decomposition, which is corresponding to the decomposition of norm. Actually we have:
\begin{itemize}
\item[3.] (Archimedean)$|x|^2=|\pi(x)|^2+|\pi^{\perp}(x)|^2 \forall x \in X$
\item[3*.](non-Archimedean)$|x|=max(|\pi(x)|,|\pi^{\perp}(x)|)\forall x \in X$
\end{itemize}
The Archimedean orthogonality property corresponds to the additivity of probability. The non-Archimedean orthogonality property means the probability of event is equal to the maximum of every event. Does it make sense to measure some $p-adic$ or non-Archimedean probabilities or observables?
\\ \\
Condition 4 says the structure of Hermite conjugate is corresponding to the structure of mod $p$ reduction and lifting. We will focus on Condition 1 and Condition 4.
\begin{defn}
We say a projection $\pi$ of $A$ is a \textbf{Orthogonal projection} if one of these condition is satisfied.
\end{defn}
Now we consider the Frobenius map $\sigma$ action over $A_0$, $A_0=\left \{ x \in A,|x|\le 1 \right \}$,$\sigma :A_0 \longrightarrow A_0,x\mapsto x^p$. Assume $A$ is a $\Q_p$-commutative Banach algebra with a unit, such that $|A|=|\Q_p|$ and $|p|=\frac{1}{p}$.
This map reduce to $\tilde{A}$ would be the usual Frobenius map of $\mathbb{F}_p$-Algebra. The orbit of Frobenius map $\sigma$ action has several classification:\\
1.topological nilpotent
$$Nil(A)= \left \{x \in A_0\ ,|x|=1,  \lim_{n \to \infty}\sigma^n(x)=0    \right \} $$
2.periodic element
$$T(A)=\left \{ x \in A_0,\exists k,\sigma^k(x)=x   \right \}\bigcup \left \{0 \right \} $$
3.chaos element
$$Chaos(A)=\left \{ x \in A_0,\forall k,\sigma^k(x)\ne x ,|\sigma^k(x)|=1 \right \} $$
\begin{thm}
The mod $p$ reduction induce a correspondence:\\
1.topological nilpotent of $A$ corresponds to nilpotent elements in $\tilde{A}$
$$Nil(A)\overset{surjective}{\longrightarrow }  Nil(\tilde{A})$$
2.periodic element of $A$ corresponds to periodic element in $\tilde{A}$
$$T(A)\overset{1:1}{\longleftrightarrow } T(\tilde{A})$$
3.chaos element in $A$ remove the quasi-periodic element corresponds to chaos element in $\tilde{A}$
\end{thm}
The proof of theorem is a standard method in theory of Witt ring\cite{Witt}. If $x,y\in T(\tilde{A})$ commutes,it's easy to check:$x+y,xy\in T(\tilde{A})$.
\begin{defn}
We call $x$ is a Teichmüller element if $x \in T(A)$. If $x$ acts on a ultrametric Banach space,we call $x$ has $\infty$ period if:
$$s-\lim_{N\to \infty} \sigma^{N!}(x)= x,|x|=1$$
\end{defn}

\begin{rmk}
The importance of Teichmüller element is that Teichmüller element has a unique lifting property and it has a nature orthogonal spectrum decomposition which can be used to build up $p-adic$ Hermite elements.\\ \\
When we are considering a spectrum theory of operator, we should based on an algebraic closed field. At least we should consider $\bar{\mathbb{F}}_p$. From Galois theory we have:
$$\sigma^N(\lambda)=\lambda^{p^N}=\lambda \Leftrightarrow \lambda \in \mathbb{F}_{p^N}$$
So $\sigma^N(x)=x$ infers that the spectrum of $x$ lies in $\mathbb{F}_{p^N}$ (We have jumped the step of mod $p$ reduction because the correspondence of periodic element is 1:1). If we want to get spectral theory of algebra over $\mathbb{F}_{p}$, we should think about the action of Galois group on the spectrum. We have:
$$\mathrm{Gal}({\mathbb{F}_{p^N}}|\mathbb{F}_p)\simeq \mathbb{Z}/N\mathbb{Z}$$
$$\mathrm{Gal}(\bar{\mathbb{F}}_p|\mathbb{F}_p)\simeq \lim_{\underset{N\ge 1}\longleftarrow}\mathbb{Z}/N!\mathbb{Z}\simeq \widehat{\mathbb{Z}}$$
$\widehat{\mathbb{Z}}$ is the pro-finite completion of $\mathbb{Z}$, $\mathbb{Z}$ is dense in $\widehat{\mathbb{Z}}$.  Frobenius map $\sigma$ is the topological generator of $\widehat{\mathbb{Z}}$.
\end{rmk}

\begin{rmk}
We will introduce the simplest case of spectrum decomposition. The base field $k$ is arbitrary.\\
Assume a linear transformation $A$ acts on finite dimensional linear space $V$. Then $A$ generates a $k$-Algebra $k[A]$. $k[A]$ is isomorphic to $k[X]/(f(X))$, $f(X)$ is the minimal polynomial of $A$. $f(X)$ has a irreducible polynomial decomposition:
$$f(X)=p_{1}(X)^{n_{1}}...p_{k}(X)^{n_{k}}$$
Let $h_i(X)=f(X)/p_{i}(X)^{n_i}$, hence $h_1(X),h_2(X)...h_k(X)$ are coprime.\\
So there exists $c_i(X)$ such that:
$$\sum_{i=1}^{k} c_i(X)h_i(X)=1$$
Let:
$$ c_i(A)h_i(A)=\pi_i$$
It's easy to show that:
$$\sum_{i=1}^{k}\pi_i=1,\ \  \pi_i^2=\pi_i,\ \ \pi_i\pi_j=0(\forall i \ne j)$$
So all $\pi_i$ are projections and the product between them are 0.
Let:
$$A_i=A\pi_i\ \ V_i=\pi_i(V)$$
then we have:
$$
1=\sum_{i=1}^{k}\pi_i,\ \ V=\oplus_{i=1}^kV_i 
$$
$$
A=\sum_{i=1}^{k}A_i,\ \  A\mid _{V_{i}}=A_i
$$
The first equation says the linear space is decomposed by $\pi_i$. The second equation says the linear transform is decomposed by $\pi_i$. There is a one to one correspondence between $\pi_i$ and points in $\mathrm{Spec}A$. $\mathrm{Spec}A$ is the prime spectrum of $k[A]$. In general we have:
$$
V=\oplus_{p \in \mathrm{Spec}A}V_p
$$
$$
k[A]=\oplus_{p \in \mathrm{Spec}A}k[A_p]
$$
We can see that the space will be decomposed by the parameterization of spectrum. So we call this formula spectrum decomposition.\\
The statements above can be generalized to Euclidean domain, Principle ideal domain, Dedekind domain. If we consider the operator with metric and measure structure, then we will get theorems like Hermite operator spectrum decomposition theorem, unitary operator spectrum decomposition theorem.\\
\end{rmk}
To simplify issues,now we let $x$ be a Teichmüller element with period 1. So we have:
$$x^{p}=x$$
If the period is not 1, the theory also works. Because we can tensor a unramified field extension $K$ of $\Q_{p}$ such that the polynomial $f(x)=x^{p^N}-x$ splits.\\
Let $\omega_i,i=1,2...p$ be the root of $\lambda^{p}=\lambda$ in $\Z_p$. Such roots exists and the amount of roots is $p$, which can be given by Teichmüller lift. Now we have:
$$\prod_{i=1}^{p}(x-\omega _i)=x^p-x$$
So we have projections and spectrum decomposition:
$$
\pi_k=\frac{\prod_{i\ne k}^{p}(x-\omega _i)}{\prod_{i\ne k}^{p}(\omega _k-\omega _i)},k=1,2,...p
$$
$$\sum_{i=1}^{p}\pi_i=1,\ \ \sum_{i=1}^{p}\omega_i\pi_i=x,\ \ \pi_i^2=\pi_i,\ \ \pi_i\pi_j=0(\forall i \ne j)$$
The projections may be 0. If one of them is not 0, then it must be an orthogonal projection.
\begin{proof}
$$
\left | \pi_k \right |=\left | \frac{\prod_{i\ne k}^{p}(x-\omega _i)}{\prod_{i\ne k}^{p}(\omega _k-\omega _i)}  \right |\le \frac{1} {|\prod_{i\ne k}^{p}(\omega _k-\omega _i)|} 
$$
The product of denominator can do mod $p$ reduction. We can view as a product in finite field.
$$
\prod_{i\ne k}^{p}(\omega _k-\omega _i)\equiv \prod_{a \in \mathbb{F}_p^*}a \equiv -1(\bmod p)
$$
On the one hand, we know:
$$
|\pi_k|\le 1
$$
On the other hand, we know:
$$
\pi_k^2=\pi_k
$$
So if $\pi_k \ne 0$, then we have:
$$
|\pi_k|^2 \ge |\pi_k| \ge 1 \Longrightarrow |\pi_k|=1
$$
\end{proof}

Now we want to show that Teichmüller element of period $\infty$ has a spectral decomposition, parameterized by $\bar{\mathbb{F}}_p$. Suppose  $X$ is ultrametric Banach space over $\widehat{\mathbb{Q}_p^{unram}}$, $\widehat{\mathbb{Q}_p^{unram}}$ is the completion of maximal unramified extension of $\Q_p$. Let $\mathcal{O}=\left \{ x\in K, \left | x \right |\le 1 \right \}$ be the integer ring of $\widehat{\mathbb{Q}_p^{unram}}$ . Suppose the range of norm satisfy:$|X|=|\widehat{\mathbb{Q}_p^{unram}}|=p^{\Z}\cup \left \{  0\right \} $.
\begin{thm}
Let $x$ be a linear operator on $X$. $x$ is a Teichmüller element of period $\infty$ if and only if there exists a spectral decomposition:
$$
\sum_{\lambda\in T(K)}\pi_\lambda=1,\ \ \sum_{\lambda\in T(K)}\lambda\pi_\lambda=x,\ \ \pi_\lambda^2=\pi_\lambda,\ \ \pi_\lambda\pi_{\lambda^*}=0(\forall \lambda \ne \lambda^*)
$$
$T(K)=\bigcup_{N=1}^{\infty} T_{N!}(K)$ is the union of Teichmüller element in $K$.
The sum converges in strong operator topology.
\end{thm}
\begin{proof}
Let $X_0=\left \{t \in X,\left | t \right |\le  1  \right \}$,$X_k=X_0/p^kX_0,k=1,2...$.Let $\phi_k:X_0 \to X_k$ be the canonical mod$p^k$ reduction,$\phi_k$ contracts the ball of radius less than $p^{-k}$ to a point in $X_0$,$\phi_k$ is surjective.Let $\phi_{jk}:X_j \to X_k,j\ge k$ be the transition morphism.\\
$X_0$ is $p-adic$ complete.So we have:
$$
\lim_{\underset{k\ge 1}\longleftarrow}X_k = X_0 
$$
Equivalently, we can write the elements in $X_0$ as:$t=(t_1,t_2,...), t_j=\phi_j(t)$, $\phi_{jk}(t_j)=t_k$, Let $x$ be a Teichmüller element of period $\infty$. $T(\mathbb{F}_{p^{n!}})$ is defined as the Teichmüller lift of $\mathbb{F}_{p^{n!}}$ in $\mathcal{O}/p^{k}\mathcal{O},k=1,2...$
We have:
$$\sigma^{N!}(x)\overset{s-lim}\longrightarrow x,|x|=1$$
By the definition of strong limit, we have:
$$\lim_{N\to \infty} \sigma^{N!}(x)(t)=x(t),\forall t \in X_0$$
Let $x_k$ be the $k-th$ reduction of $x$. We have:$x_k:X_k \to X_k$
$$\lim_{N\to \infty} \sigma^{N!}(x_k)(t_k)=x_k(t_k),\forall t_k \in X_k$$
Since the ball of radius less than $p^{-k}$ have contracted to a point in $X_k$, we have:
$$\forall t_k \in X_k,\exists n,\sigma^{n!}(x)(t_k)=x(t_k)$$
So we have:$X_k=\bigcup_{n\ge1}X_{k,n}$,$X_{k,n}=\left \{t_k \in X_k ,\sigma^{n!}(x_k)(t_k)=x_k(t_k)\right \} $,$X_{k,n}\subseteq X_{k,n^*}$ when $n \le n^{*}$. We define:
$$\pi_{k,\lambda,n}=\prod_{\lambda^* \ne \lambda,\lambda^*\in T(\mathbb{F}_{p^{n!}})}\frac{x_k-\lambda^*}{\lambda-\lambda^*} $$
$\pi_{k,\lambda,n}$ is the projection of $X_{k,n}$. In general we have:
$$
\sum_{\lambda\in T(\mathbb{F}_{p^{n!}})}\lambda \pi_{k,\lambda,n}=1\mid_{X_{k,\lambda,n}},\sum_{\lambda \in T(\mathbb{F}_{p^{n!}})}\lambda  \pi_{k,\lambda,n}=x_k\mid_{X_{k,\lambda,n}} , \ \  \pi_{k,\lambda,n}^2=\pi_{k,\lambda,n}, \pi_{k,\lambda,n}\pi_{k,\lambda^*,n}=0(\forall \lambda \ne \lambda^*)\ \
$$
We can get a direct sum:
$$X_{k,n}=\underset{\lambda  \in T(\mathbb{F}_{p^{n!}})}{\oplus }X_{k,\lambda},\forall n \ge 1$$
Since $X_k=\bigcup_{n\ge1}X_{k,n}$, We have:
$$X_k=\underset{\lambda  \in T(\bar{\mathbb{F}}_{p})}{\oplus }X_{k,\lambda}$$
We get the spectral decomposition of $x_k$ over $X_k$. Let $\pi_{k,\lambda}$ be the projection: $\pi_{k,\lambda}:X_k \to X_{k,\lambda}$.
It's no hard to see that:
$$x_k=\sum_{\lambda \in  T(\bar{\mathbb{F}}_{p})}\lambda\pi_{k,\lambda}$$
We want to show:
$$\lim_{\underset{k\ge 1}\longleftarrow}X_{k,\lambda} = X_{\lambda},X_0=\underset{\lambda\in T(K)}{\hat{\oplus }}X_{\lambda}$$
Here we define $X_{\lambda}=\left \{ t \in X_0, x(t)=\lambda t\right \} $. $\lambda$ is a element in $T(K)$.
We get two commutative diagrams:
\\
\begin{tikzcd}
X_{j} \arrow["\pi_{j,\lambda}"',d,two heads]  \arrow["\phi_{jk}"',r,two heads] & X_{k} \arrow["\pi_{k,\lambda}"',d,two heads] & X_{\lambda} \arrow[d] \arrow[rd] \arrow["\exists ! \theta"',r,dashed] & \underset{\underset{k\ge 1}\longleftarrow}{\lim}X_{k,\lambda} \arrow["\phi_{j}"',ld,two heads] \arrow["\phi_{k}"',d,two heads]  
\\
X_{j,\lambda} \arrow["\phi_{jk}"',r,two heads] & X_{k,\lambda} & X_{j,\lambda} \arrow["\phi_{jk}"',r,two heads] & X_{k,\lambda}
\end{tikzcd}
\\
We want to show $\theta$ is isomorphism. Let:
$$
\lambda=(\lambda_1,\lambda_2,...) \in \mathcal{O}=\underset{\underset{k\ge 1}\longleftarrow}{\lim}\mathcal{O}/p^k\mathcal{O},
t=(t_1,t_2,...) \in \underset{\underset{k\ge 1}\longleftarrow}{\lim}X_{k,\lambda}
$$
$$
x_k(t_k)=\lambda_k t_k,\phi_k(t_j)=t_k,\phi_k(\lambda_j)=\lambda_k,j \ge k
$$
We have:
$$x(t)=(\lambda_1t_1,\lambda_2t_2,...),\phi_k(\lambda_{j}t_{j})=\lambda_{k}t_{k}$$
Hence it has a limit:$x(t)=\lambda t$. So
$\pi_{\lambda}=\underset{\underset{k\ge 1}\longleftarrow}{\lim}\pi_{k,\lambda}:X_0 \to X_{\lambda}$ is a well-defined orthogonal projection. We only need to check the completeness of $ \left \{ \pi_{\lambda} \right \}  $,i.e:
$$
X_0=\underset{\lambda\in T(K)}{\hat{\oplus }}X_{\lambda}
$$
Considering the mod$p^k$ reduction:
$$
\phi_k:X_0 \to X_k=\underset{\lambda  \in T(\bar{\mathbb{F}}_{p})}{\oplus }X_{k,\lambda}
$$
$$
\phi_k:X_{\lambda} \to X_{k,\lambda}
$$
Since the direct sum is always finite, mod$p^k$ reduction is surjective. We get:$\oplus_{\lambda} X_{\lambda }$ is dense in $X_0$. We decompose the Teichmüller element $x$ of period $\infty$:
$$
\sum_{\lambda\in T(K)}\pi_\lambda=1,\ \ \sum_{\lambda\in T(K)}\lambda\pi_\lambda=x,\ \ \pi_\lambda^2=\pi_\lambda,\ \ \pi_\lambda\pi_{\lambda^*}=0(\forall \lambda \ne \lambda^*)
$$
Finally, if $x$ has such a spectral decomposition(in strong operator topology), it's easy to check $x$ is a Teichmüller element.
\end{proof}

Now we will think about the reduction of orthogonal projection family and the lifting of projection family. In fact there exists 1:1 correspondence. Recall our notation:$A$ is a $K$-ultrametric Banach algebra such that: $\left | A \right | =p^{\mathbb{Z}}\cup \left \{ 0 \right \} $, $K$ is a finite unramified extension of $\Q_p$. Let $A_0$ be the unit ball of $A$:$A_0=\left \{ x \in A,|x|\le 1 \right \}$. $A_0$ is a nature $\Z_p$-Algebra which can be reduced to $\tilde{A}=A_0/pA_0$. $\tilde{A}$ has a nature $\mathbb{F}_p$-Algebra structure.\\
\begin{lem}(Unique Lifting Lemma)\\
Suppose $\bar{\pi}$ is a projection in $\tilde{A}$. There exists a unique lifting of $\bar{\pi}$, we call it $\pi$. $\pi$ is orthogonal projection, $\sigma(\pi)=\pi$.\\
Suppose $\bar{\pi}_k,k=1,2...p$ is projections in $\tilde{A}$. We have:
$$\sum_{i=1}^{p}\bar{\pi}_i=1,\ \ \bar{\pi}_i^2=\bar{\pi}_i,\ \ \bar{\pi}_i\bar{\pi}_j=0(\forall i \ne j)$$
Suppose $\pi_k, k=1,2,...p$ is the unique lifting of $\bar{\pi}_k$, then the lifting also satisfy the relationship:
$$\sum_{i=1}^{p}\pi_i=1,\ \  \pi_i^2=\pi_i,\ \ \pi_i\pi_j=0(\forall i \ne j)$$
\end{lem}
\begin{proof}
Suppose we select a lifting of $\bar{\pi}$ named $a$, such that:$a^2=a+py,y \in A_0$. Considering the Cauchy sequence:$\left \{ \sigma^k(a),k=1,2,... \right \} $. The $\sigma$ is Frobenius map. Then there exists a unique limit $\pi$ such that:$\sigma(\pi)=\pi,\pi^2=\pi$. We have $\pi$ is orthogonal projection.\\
Suppose we select liftings of $\bar{\pi_k}$ named $\pi_k$. First, we have estimate:
$$|\pi_i\pi_j|\le \frac{1}{p} (\forall i \ne j)$$
Secondly, suppose $\pi_i\pi_j \ne 0$, we have (This estimate depends on the commutative of $\pi_i$):
$$|\pi_i\pi_j|\ge 1 (\forall i \ne j)$$
which leads to a contradiction.
Finally, $\left ( \sum_{i=1}^{k}\pi_i \right )^2=\sum_{i=1}^{k}\pi_i$ and $\sum_{i=1}^{k}\pi_i$ is in the neighborhood of 1, so $\sum_{i=1}^{k}\pi_i=1$.
\end{proof}

\section{Spectral measure of $p-adic$ Hermite elements}
Now let's recall the defnition of $p-adic$ Hermite elements.\\
Let $K$ be an unramified extension of $\Q_p$ such that $\omega^{p^N}=\omega$ has exactly $p^N$ solutions. Let $A$ be a commutative $K$-Banach algebra with unit, $\sigma$ be the Frobenius map:$\sigma:A \to A ,x\mapsto x^p$. We assume the norm on $A$ satisfy:$\left | A\right | =\left |K\right | $. Let $A_0$ be the unit ball: $A_0=\left \{ x\in A, \left | x \right |\le 1 \right \} $. 
\begin{defn}
Let $x \in A$, We say $x$ is a $p-adic$ hermite operator with period $N$, if $x$ has a taylor expansion:
$$x=\sum_{i=k}^{\infty} x_ip^i,x_i \in T_N(A_0)$$
\end{defn}
Let $k=0,N=1$. Then we have:$x_i^p=x_i,i=0,1,2...$, the assumption $k=0$ is convenient for our discussion.\\

The following lemma is not clearly point out before. However,it is obvious. 
\begin{lem}
If $\pi$,$\pi^*$ are orthogonal projections, and they are commute. Then $\pi\pi^*$ are either orthogonal projection or 0.
\end{lem}

We will show that how to define a spectral measure of $x$\\
$x_i$ satisfies a polynomial equation:$x_i^p=x_i$. So we can define a spectral decompositon of $x_i$.We have:
$$
\exists \pi _{i,j},j=1,2..p \ \ \  \pi _{i,j}^2=\pi _{i,j}\ \ \pi _{i,j}\pi _{i,k}=0(\forall j \ne k)
$$
$$
\sum_{j=1}^{p} \pi_{i,j}=I \ \ 
\sum_{j=1}^{p} \omega _j\pi_{i,j}=x_i
$$
$\omega _j$ is the root of $\lambda^p=\lambda $.
Since $|x_i|=1$ (if $x_i=0$, the same argument holds) and the norm over $A$ has the Non-Archimedean property, it's not hard to prove that:
$$|\pi_{i,j}|=1\ or \ 0$$
So the spectral decompositon of $x_i$ is an orthogonal spectral decompositon.\\
We can apply this spectral decompositon step by step to get the spectral decompositon of $x$.\\
Actually we have:
$$
I=\sum_{j=1}^{p} \pi_{0,j}
$$
$$
I=\sum_{j,k=1}^{p} \pi_{0,j} \pi_{1,k}
$$
$$
...
$$
$$
I=\lim_{n \to \infty} \sum_{i_0,i_1,...i_n=1}^{p} \pi_{0,i_0}\pi_{1,i_1}...\pi_{n,i_n}=\int_{\Z_p}dE_{\lambda }
$$
$$
x_0=\sum_{j=1}^{p} \omega _j\pi_{0,j}
$$
$$
x_0+px_1=\sum_{j,k=1}^{p}(\omega _j+p\omega _k) \pi_{0,j} \pi_{1,k}
$$
$$
......
$$
$$
x=\lim_{n \to \infty} \sum_{i_0,i_1,...i_n=1}^{p}(\omega _{i_0}+p\omega _{i_1}+...+p^{n}\omega _{i_n}) \pi_{0,i_0}\pi_{1,i_1}...\pi_{n,i_n}=\int_{\Z_p}\lambda dE_{\lambda }
$$
$$
\begin{matrix}
  \pi_{0,1}&\pi_{0,2} &...&\pi_{0,p}& \\
  \pi_{1,1}&\pi_{1,2}&...  &\pi_{1,p}& \\
  \pi_{2,1}&\pi_{2,2}&... &\pi_{2,p} &  \\
 ...&...&... &... &
\end{matrix}
$$
For each path from top to bottom on the diagram corresponding to a unique number in $\Z_p$. Counting from top to bottom, the projections on the first floor correspond to the canonical resolution of $\Z_p$ to small discs of whose radius is $r=\frac{1}{p}$. The second floor correspond to the resolution of small discs of $\Z_p$ whose radius is $r=\frac{1}{p^2}$ and so on.\\ \\
In this sense we build up the \textbf{fractal spectral measure} of Hermite operator $x$.It is a finite additive Orthogonal projection valued measure over $\Z_p$.\\
In general,without the assumption of $k=0$,we have:
\begin{thm}($p-adic$ Hermite operator Spectral decomposition theorem)\\
Let $x$ be a $p-adic$ Hermite operator, there exists a spectral integral:
$$
I=\int_{\Q_p}dE_{\lambda }  \ \ \   x=\int_{\Q_p}\lambda dE_{\lambda }
$$
\end{thm}
We can compare the result with the Archimedean case:
\begin{thm}(Hermite operator spectral decomposition theorem)\\
Let $x$ be a Hermite operator, there exists a spectral integral:
$$
I=\int_{\R}dE_{\lambda }  \ \ \  x=\int_{\R}\lambda dE_{\lambda }
$$
\end{thm}
There is a nature connection with the path-intergral in quantum mechanics. As we know, the original definition of the path-intergral is to apply the spectral decomposition theorem in different time again and again to calculate the propagator. Our definition is related to this statement. There is a one to one correspondence between the path in the following diagram to the filtration of balls in $\Z_p$. We should sum the all paths in the diagram to give the spectral decomposition.
$$
\begin{matrix}
  \pi_{0,1}&\pi_{0,2} &...&\pi_{0,p}& \\
  \pi_{1,1}&\pi_{1,2}&...  &\pi_{1,p}& \\
  \pi_{2,1}&\pi_{2,2}&... &\pi_{2,p} &  \\
 ...&...&... &... &
\end{matrix}
$$
What happens in finite dimension case? Let $X$ be a $\Q_p$-linear space with finite dimension,  equipped with a non-Archimedean norm such that $|X|=|\Q_p|=p^{\Z}\cup \left \{ 0 \right \} $. Let $x$ be a $p-adic$ Hermite operator acts on $X$. Let $\varPi_n=\left \{  \pi_{0,i_0}\pi_{1,i_1}...\pi_{n,i_n} \right \}_{i_k=1,2...p}-\left \{0 \right \}$, $\varPi_n$ is the spectral decomposition of $x_0+px_1+...+p^nx_n$, which is complete and orthogonal. The element in $\varPi_{n+1}$ is the subprojection of $\varPi_n$. It's obvious to see:
$$Card(\varPi_n)\le Card(\varPi_{n+1})\le dimX$$
Hence it has a limit:
$$\lim_{n \to \infty} Card(\varPi_n)$$
Hence we have: $\varPi_n$ is stable for some $n_0$. Let $\varPi$ be the limit set of $\varPi_n$.$\varPi$ is composed of a finite number of Orthogonal projections. Actually, we have:
$$x=\sum_{i=1}^{Card(\varPi)}\lambda_i\pi_i,\lambda_i \in \Q_p$$
$$\sum_{i=1}^{Card(\varPi)}\pi_i=1,\ \  \pi_i^2=\pi_i,\ \ \pi_i\pi_j=0(\forall i \ne j)$$
We will give a more detailed description immediately.

\section{Examples}
In this section, We want to give some examples of $p-adic$ Hermite operator.\\
Let $\Q_p^n$ equipped with supremum norm:
$$\forall x \in \Q_p^n,x=(x_1,x_2...,x_n),x_i \in \Q_p,|x|=\sup_{1\le i \le n}   \left (|x_i|\right )$$
Here $n$ is an arbitrary integer. Let $A$ be a linear operator acts on $\Q_p^n$. The following conditions are equivalent:
\begin{enumerate}
\item $A$ is a $p-adic$ Hermite operator.
\item There exists a orthogonal projection spectral decomposition:
$$\sum_{\lambda}\pi_\lambda=1,\ \ \sum_{\lambda}\lambda\pi_\lambda=A,\ \ \pi_\lambda^2=\pi_\lambda,\ \ \pi_\lambda\pi_{\lambda^*}=0(\forall \lambda \ne \lambda^*)
$$
$\lambda \in \Q_p$ is the eigenvalue of $A$.
\item $\exists U \in GL_n(\Z_p),A=Udiag(\lambda_1,...\lambda_n)U^{-1}$
\end{enumerate}
\begin{proof}
$1\Rightarrow 2$:We have shown.\\ \\ 
$2\Rightarrow 3$:The columns of elements in $GL_n(\Z_p)$ is $p-adic$ orthonormalized. We mean for any $g \in GL_n(\Z_p)$, write $g$ as a column vector combination:$g=(X_1,X_2...,X_n)$. We have:
$$\forall (c_1,c_2...,c_n) \in \Q_p^n,\left |\sum_{i=1}^{n} c_i X_i \right |=\sup_{1\le i\le m}\left |c_i  \right | $$
The converse proposition is also true. Actually there is a one to one correspondence between $p-adic$ orthonormalized column vectors with elements in $GL_n(\Z_p)$. $GL_n(\Z_p)$ is the $p-adic$ substitute of the Archimedean case orthogonal group:$O_n(\R)$ or $U_n(\CC)$. We can find a orthonormal basis of $\Q_p^n$ corresponds with the orthogonal projection spectral decomposition to get $U$.\\ \\
$3\Rightarrow 1$: We have the Taylor expansion of $diag(\lambda_1,...\lambda_n)$ in $\Q_p^n$. So it is a $p-adic$ Hermite operator. Moreover, it's easy to see $A$ is a $p-adic$ Hermite operator if and only if $U^{-1}AU$ is a $p-adic$ Hermite operator.
\end{proof}
Let $K$ be a field, $\mathrm{Char}K=0$. Let $K[X_1,X_2,...X_n]$ be the polynomials with $n$ variables over $K$.\\ \\
For homogeneous polynomials, we have Euler Theorem:
\begin{defn}
We call $f(X_1,X_2,...X_n)\in K[X_1,X_2,...X_n]$ is a homogeneous polynomial of degree $k$ if:
$$f(\alpha X_1,\alpha X_2,...\alpha X_n)=\alpha^kf(X_1,X_2,...X_n),\forall \alpha \in K^{\times}$$
\end{defn}

\begin{thm}
We call $\Delta=\sum_{i=1}^{n} X_i\frac{\partial}{\partial X_i}$  \textbf{Euler operator}, $f(X_1,X_2,...X_n)\in K[X_1,X_2,...X_n]$ is a homogeneous polynomial of degree $k$ if and only if:$\Delta f=kf$
\end{thm}
Let $K<X_1,X_2,...X_n>$ be the Tate Algebra:
$$
K<X_1,X_2,...X_n>=\left \{ f= \sum_{i_1,i_2,...i_n}^{\infty} a_{i_1,i_2,...i_n}X_1^{i_1}X_2^{i_2}...X_n^{i_n},a_{i_1,i_2,...i_n}\to 0\right \}
$$
with respect to the Gauss norm:
$$
\left | \sum_{i_1,i_2,...i_n}^{\infty} a_{i_1,i_2,...i_n}X_1^{i_1}X_2^{i_2}...X_n^{i_n} \right |=\sup\left |a_{i_1,i_2,...i_n} \right | 
$$
Here we assume $K$ is a field extension of $\Q_p$ with non-Archimedean norm.
\begin{prop}
The $\textbf{Euler operator}$ $\Delta$ is a $p-adic$ Hermite operator over $K<X_1,X_2,...X_n>$
with eigenvalue: $k=0,1,2,...$ and eigenfunction: $f_0=1;f_{1,1}=X_1,f_{1,2}=X_2...f_{1,n}=X_n;...$. Eigenfunctions are $p-adic$ orthonormalized. In general we can define the creation and annihilation operators:
$$
a^{+}_{i}=X_i,a^{-}_{i}=\frac{\partial}{\partial X_i},[a_{i}^{-},a_{i}^{+}]=1,\Delta=\sum_{i=1}^{n}a^{+}_{i}a^{-}_{i}
$$
\end{prop}
Now we will simplify examples given by A.N.Kochubei. We will follow on \cite{NANO} and try to establish a general framework. \\
Let $C(\Z_p,\CC_p)$ be the set of continous function from $\Z_p$ to $\CC_p$. We define:
$$(a^{+}f)(x)=xf(x-1),(a^{-}f)(x)=f(x+1)-f(x),x \in \Z_p$$
$a^{+},a^{-}$ are bounded and satify the relation $[a^{-},a^{+}]=1$. Let $A=a^{+}a^{-}$, then $A$ is $p-adic$ hermite operator with eigenfunction:
$$
P_n(x)=\frac{x(x-1)...(x-n+1)}{n!},n \ge 1; P_0(x)=1
$$
such that $AP_n=nP_n$.
Let $(a^{*}f)(x)=f(x+1)$, We have:$[a^{*},a^{+}]=1$.Let's define:
$$
A^{*}=a^{+}a^{*},(A^{*}f)(x)=xf(x)
$$
$A^{*}$ is $p-adic$ hermite operator because $A^{*}=x$ is a element in $C(\Z_p,\Z_p)$ , which has the  property:
$$
\forall f(x)\in C(\Z_p,\Z_p),\sigma(f(x))=(f(x))^p\equiv f(x) \pmod{p}
$$
We want to summarize the examples listed above. Let $X$ be a $K$-ultrametric Banach space($K$ is extension of $\Q_p$) or Hilbert space over $\CC$. In quantum mechanics and quantum field theory we know the creation operators can express particle production. Here is our definition:
\begin{defn}
We call a triple $(\Omega,\hat{a}^{+},\hat{a}^{-})$ creation and annihilation operators if they satisfy the following conditions($\Omega \in X$ is a special element):
\begin{itemize}
\item[0.] We call $\Omega$ \textbf{vacuum} or \textbf{ground state}, $\hat{a}^{+}$ \textbf{creation operator}, $\hat{a}^{-}$ \textbf{annihilation operator}.
\item[1.] $\left \{ (\hat{a}^{+})^{n}(\Omega)\right \} _{n=0,1,...}$ generates $X$. They are orthogonal to each other.
\item[2.] $\hat{a}^{-}(\Omega)=0$, $[\hat{a}^{-},\hat{a}^{+}]=1$.
\end{itemize}
\end{defn}
We will express the condition 1 more accurately.
\begin{itemize}
\item[1.] (Archimedean)$$\overline{Span(\Omega,\hat{a}^{+}(\Omega),(\hat{a}^{+})^{2}(\Omega)...)}=X$$
$$\forall n \ne m,((\hat{a}^{+})^{n}(\Omega),(\hat{a}^{+})^{m}(\Omega))=0$$
\item[1*.](non-Archimedean)$$\overline{Span(\Omega,\hat{a}^{+}(\Omega),(\hat{a}^{+})^{2}(\Omega)...)}=X$$
$$\left | \sum_{n=0}^{k} c_n (\hat{a}^{+})^{n}(\Omega)\right | =\underset{0\le n\le k}{\sup}\left | c_n \right |\left | (\hat{a}^{+})^{n}(\Omega) \right | |, k=0,1,2...$$
\end{itemize}
\begin{prop}
In the conditions above, we have: $N=\hat{a}^{+}\hat{a}^{-}$ is a Hermite operator(usual hermite operator if $X$ is a Hilbert space, $p-adic$ hermite operator if $X$ is a $K$-ultrametric Banach space) with the spectrum:$SpecN =\left \{ 0,1,2,... \right \} $. $\Omega$ is the cyclic vector.
\end{prop}
\begin{rmk}
In Quantum mechanics, we know the hamilton of \textbf{Harmonic oscillator}:
$$H=\frac{1}{2} m\omega^2x^2-\frac{1}{2m} \frac{d^2 }{d x^2} $$
It is well-known that:
$$
\hat{a}^+=\frac{1}{\sqrt{2m\omega} }(m\omega x-\frac{d}{dx} ),\hat{a}^-=\frac{1}{\sqrt{2m\omega} }(m\omega x+\frac{d}{dx} )
$$
$$
\Omega=e^{-\frac{m\omega x^2}{2}},H=\omega(\hat{a}^+\hat{a}^{-}+\frac{1}{2})
$$
\end{rmk}
Let $X=K<X>$ be the Tate Algebra, we define:
$$
\hat{a}^+=x,\hat{a}^-=\frac{d}{dx} 
$$
$$
\Omega=1,N=\hat{a}^+\hat{a}^{-}
$$
There is not a unique choice. We can also define:
$$
\hat{a}^+=x+h(\frac{d}{dx} ),\hat{a}^-=\frac{d}{dx} ;h(X) \in K<X>,\left | h(X) \right |\le1
$$
$$
\Omega=1,N=\hat{a}^+\hat{a}^{-}
$$
The more accurately \textbf{vacuum} of the Tate Algebra is not identity function. We should view it as $1_{\mathcal O}$ since the Tate Algebra describe the analytic geometry of $\mathcal O$. The statement listed above coincide with Tate's thesis\cite{Tate}. We have the correspondence of fast decreasing functions:
$$
1_{\Z_p} \longleftrightarrow e^{-\frac{x^2}{2}},x \in \R
$$
We will show there exists a decomposition of $n \times  n$ $p-adic$ matrix $A \in M_n(\mathcal O_K)$, $K=\widehat{\mathbb{Q}_p^{unram}}$. Consider the $modp$ reduction: $M_n(\mathcal O_K) \to M_n(\bar{\mathbb{F}}_p)$\\ \\
$M_n(\bar{\mathbb{F}}_p)$ is a non-commutative $\mathbb{F}_p$-algebra. Let $A \in M_n(\mathcal O_K)$ be a matrix. We can define a norm: $A=(a_{ij})_{1\le i,j \le n},|A|=\underset{1 \le i,j\le n}{\sup }|a_{ij}|$. Let $\mathcal{A}$ be the $K$-algebra generated by $A$. Let $\widetilde{A} \in M_n(\bar{\mathbb{F}}_p)$ be the reduction of $A$. The coefficient of $\widetilde{A}$ can be embedded into a common finite extension of $\mathbb{F}_p$. So $\tilde{A}$ generates a finite dimensional commutative $\mathbb{F}_p$-algebra $\widetilde{\mathcal{A}}$.
\begin{prop}
Let $\widetilde{\mathcal{A}}$ be a finite dimensional commutative $\mathbb{F}_p$-algebra. Let $\widetilde{X} \in \widetilde{\mathcal{A}}$ be a arbitrary element. Let $\sigma$ be the Frobenius morphism. There exists a canonical Jordan decomposition:
$$\widetilde{X} =\widetilde{X_s}+\widetilde{X_n},\widetilde{X_s}=\lim_{k \to \infty} \sigma^{k!}(\widetilde{X})$$
$\widetilde{X_s}$ is a Teichmüller element:
$$\lim_{k \to \infty} \sigma^{k!}(\widetilde{X_s})=\widetilde{X_s}$$
Moveover, $\widetilde{X_s}$ has finite period. \\
$\widetilde{X_n}$ is a nilpotent element:
$$\lim_{k \to \infty} \sigma^{k!}(\widetilde{X_n})=0$$
\end{prop}
\begin{proof}
The Frobenius morphism $\sigma$ is a $\mathbb{F}_p$-linear map on $\widetilde{\mathcal{A}}$, since $\widetilde{\mathcal{A}}$ is finite dimensional commutative $\mathbb{F}_p$-algebra. From linear algebra, $\widetilde{\mathcal{A}}$ has a Jordan decomposition:
$$\widetilde{\mathcal{A}}=\widetilde{\mathcal{A}}_s\oplus \widetilde{\mathcal{A}}_n$$
We have: $\sigma\mid_{\widetilde{\mathcal{A}}_s}$ is invertible, $\sigma\mid_{\widetilde{\mathcal{A}}_n}$ is nilpotent. Consider the cardinality of $\widetilde{\mathcal{A}}$ is finite, we have:
$$\forall x\in \widetilde{\mathcal{A}}_s,\exists k,\sigma^k(x)=x$$
\end{proof}
\begin{thm}
Let $A \in M_n(\mathcal O_K)$ be a matrix. Suppose $|A|=1$. There exists a canonical Jordan decomposition:
$$
A=A_s+A_n
$$
$A_s \in T(M_n(\mathcal O_K))$ is a Teichmüller element with finite period, $A_n$ is a topological nilpotent element such that:
$$|A_n| \le 1, \lim_{k \to \infty} \sigma^{k!}(A_n)=0$$
\end{thm}
\begin{proof}
Let $\widetilde{A}$ be the reduction of $A$, $\widetilde{A} =\widetilde{A_s}+\widetilde{A_n}$.
Let $A_s^{*}$ be a lift of $\widetilde{A_s}$ in $\mathcal{A}$, which is the K-algebra generated by $A$. $\mathcal{A}$ is commutative, so we have a limit independent of the choice of $A_s^{*}$:
$$\lim_{k \to \infty} \sigma^{k!}(A_s^*)=A_s$$
From the unique lift lemma, $A_s$ is a Teichmüller element. So we have a canonical Jordan decomposition:
$$
A=A_s+A_n,A_n=A-A_s
$$
Finally, the reduction of $A_n$ is $\widetilde{A_n}$. So we have:
$$\lim_{k \to \infty} \sigma^{k!}(\widetilde{A_n})=0 \Rightarrow  |A_n| \le 1, \lim_{k \to \infty} \sigma^{k!}(A_n)=0$$
\end{proof}

\section{$p-adic$ uncertainty principle}
In this section, we want to set up uncertainty principle for $p-adic$ Hermite operator. Let $K$ be the completion of maximal unramified extension of $\Q_p$. Let $X$ be an ultrametric Banach space over $K$ such that the range of norm $|\ \ |:X \to \R_+$ on $X$ satisfy: $|X|=|K|$. In a sense it is a formal proof in physics.
\begin{lem}
Let $A$ be a $p-adic$ Hermite operator. Let $SpecA$ be the spectrum of $A$. We have:
$$
\left | A\right | =\underset{\lambda \in SpecA}{\sup}|\lambda|
$$
\end{lem}
\begin{defn}
We define the spectrum diameter of $A$:
$$
diam(A) =\underset{\lambda ,\lambda^{*}\in SpecA}{\sup}|\lambda-\lambda^{*}| \le |A|
$$
\end{defn}
\begin{lem}
For any translation: $\widetilde{A}=A-\mu$, $\mu \in SpecA$, we have:
$$
|\widetilde{A}| =diam(A)  
$$
\end{lem}
\begin{proof}
We only prove the Lemma 5.2. The ultrametric property shows that:
$$
|\widetilde{A}|=\underset{\lambda \in SpecA}{\sup}|\lambda-\mu| \le \underset{\lambda ,\lambda^{*}\in SpecA}{\sup}|\lambda-\lambda^{*}|=diam(A)
$$
$$
diam(A)=diam(\widetilde{A}) \le |\widetilde{A}|
$$
\end{proof}
Let $A,B$ be $p-adic$ Hermite operators on $X$. Let $\psi \in X,|\psi|=1$ be the normalized wave function.
\begin{thm}($p-adic$ uncertainty principle)
We have:
$$
\left | [A,B]\psi \right |  \le diam(A)*diam(B)
$$
\end{thm}
\begin{proof}
We have:
$$
\left | [A,B]\psi \right |  \le diam(A)*diam(B)
$$
Let $\widetilde{A},\widetilde{B}$ be the translation of $A,B$ such that:
$$
|\widetilde{A}| =diam(A),|\widetilde{B}| =diam(B)
$$
We have:
$$
\left | [A,B]\psi \right |  =\left | [\widetilde{A},\widetilde{B}]\psi \right |=\left | (\widetilde{A}\widetilde{B}-\widetilde{B}\widetilde{A})\psi \right | 
\le\sup (|\widetilde{A}\widetilde{B}\psi|,|\widetilde{B}\widetilde{A}\psi|) \le |\widetilde{A}||\widetilde{B}||\psi|
$$
$$
|\widetilde{A}||\widetilde{B}||\psi|=diam(A)*diam(B)*1=diam(A)*diam(B)
$$
\end{proof}
\begin{rmk}
Recall the classical uncertainty principle:
$$
\frac{1}{2} \left |\left \langle [A,B]  \right \rangle \right |  \le \bigtriangleup A*\bigtriangleup B
$$
Here $A,B$ is the classical Archimedean Hermite operator, $\bigtriangleup A,\bigtriangleup B$ is the variance of $A,B$.
\end{rmk}
\section{Further discussion}
We refer to \cite{formal} for theory of formal group scheme.
\begin{notation}
We assume $K$ is a field extension of $\Q_p$ with a non-Archimedean norm. Let $A,B$ be the ultrametric $K$-Banach algebra(not necessary commutative) with unit. $A_0,B_0$ be the unit ball of $A,B$.
Let $\textup{Hom}_{non-Arch}(A,B)$ be the set of \textbf{Norm decreasing} $K$-Banach algebra morphism:
$$
\textup{Hom}_{non-Arch}(A,B)=\left \{f \textup{ is $K$-Banach Algebra morphism,}\forall x \in A,\left | f(x) \right | \le \left |x \right | \right \} 
$$
It's obvious to see the norm decreasing morphism $f$ satisfy: $f(A_0)\subseteq B_0$. Let $\sigma$ be the Frobenius map $\sigma:x \mapsto x^p$\\
We can define the set of non-Archimedean \textbf{Unitary operator} of $A$:
$$
U(A)=\left \{u\in A;\left |u \right |=\left |u^{-1}\right |=1\right \} 
$$
the set of non-Archimedean \textbf{orthogonal projection} of $A$:
$$
\Pi(A)=\left \{\pi\in A;\left |\pi \right |=1,\pi^2=\pi\right \} 
$$
the set of \textbf{Teichmüller element} of $A$:
$$
T(A) = \left \{x\in A;\exists k,\sigma^k(x)=x,\left |x \right | = 1\right \} 
$$
the set of \textbf{Hermite operator} of $A$:
$$
H(A)=\left \{x\in A;x=\sum_{i=0}^{\infty}x_ip^i,\sigma(x_i)=x_i,\left |x_i\right |\le 1,x_ix_j=x_jx_i \right \}
$$
For simplifying the issues, we assume the Hermite operator have a common period $1$ and bounded.
\end{notation}

\begin{prop}
The norm decreasing morphism $f \in Hom_{non-Arch}(A,B)$  preserves all the set above. We have:
$$f(U(A))\subseteq U(B),f(\Pi(A))\subseteq \Pi(B),f(T(A))\subseteq T(B),f(H(A))\subseteq H(B)$$
\end{prop}
What about the Archimedean case?\\
\begin{notation}
Let $A,B$ be the $C^*$-Algebra(not necessary commutative) with unit.\\
Let $\textup{Hom}_{Arch}(A,B)$ be the set of $C^*$-Algebra morphism. We mean:
$$
\textup{Hom}_{Arch}(A,B)=\left \{f \textup{ is $\CC$-Banach algebra morphism,}\forall x \in A,f(x^{\dagger})=f(x)^{\dagger} \right \}  
$$
It is known that the morphism between $C^*$-Algebra are all \textbf{Norm decreasing}.\\
We can define the set of \textbf{Unitary operator} of $A$:
$$
U(A)=\left \{u\in A;u^{\dagger}=u^{-1}\right \} 
$$
the set of \textbf{Orthogonal projection} of $A$:
$$
\Pi(A)=\left \{\pi \in A;\pi^{\dagger}=\pi,\pi^2=\pi\right \} 
$$
the set of \textbf{Hermite operator} of $A$:
$$
H(A)=\left \{x\in A;x^{\dagger}=x \right \}
$$
\end{notation}
\begin{prop}
The $C^*$-Algebra morphism $f \in Hom_{Arch}(A,B)$  preserves all the set above. We have:
$$f(U(A))\subseteq U(B),f(\Pi(A))\subseteq \Pi(B),f(H(A))\subseteq H(B)$$
\end{prop}
In Archimedean case($A$ is $C^*$-Algebra) we have $\forall x,y \in H(A)\Rightarrow x+y \in H(A)$. In non-Archimedean case we \textbf{will not} get the same result. However it can show that:
\begin{prop}
Suppose $x,y \in H(A)$ such that $xy=yx$, we have:$x+y,xy \in H(A)$, $A$ is $C^*$-Algebra or ultrametric Banach algebra.
\end{prop}
We define:$h=K\left [ X_0,X_1,...  \right ]/(X_0^p-X_0,X_1^p-X_1,...)$ with a norm:
$$\left |\sum_{i}a_{i_0,i_1,...}X_0^{i_0}X_1^{i_1}...\right |= \underset{\omega_i\in T(\Z_p)}{ \sup}\left | a_{i_0,i_1,...}{\omega}_0^{i_0}{\omega}_1^{i_1}...\right |$$
Since $f=\sum_{i}a_{i_0,i_1,...}X_0^{i_0}X_1^{i_1}... \in h$ is a finite sum, $\left |f \right |$ is well-defined. So we can do completion of Banach algebra: $h \to \mathcal{H}$, $\mathcal{H}$ is the "home of all $p-adic$ Hermite operator". We have:
$$
\textup{Hom}_{non-Arch}(\mathcal{H},A)\overset{1:1}{\longleftrightarrow } H(A)
$$
In Archimedean case, the "home of all Unitary operator" or "home of all Hermite operator" actually is a scheme. We know the Gelfand representation shows that the unitary operator can be realized as the continous function on $\mathrm{S}^1$, the Hermite operator can be realized as the continous function on $\R$. We have:
$$
\textup{Hom}_{Arch}(C(\mathrm{S}^1),A)\overset{1:1}{\longleftrightarrow } U(A)
$$
$$
\textup{Hom}_{Arch}(C(I),A)\overset{1:1}{\longleftrightarrow } H_{\le 1}(A),I=[-1,1]\subset \R
$$
Which is the classical result. In the final, we define:
$$
K<t,t^{-1}>=\left \{f;f=\sum_{i=-\infty}^{\infty}a_it^i ,\left | a_i\right |\to 0,i \to \infty,\left | f\right |=\sup\left |a_i \right |   \right \}
$$
We have:
$$
\textup{Hom}_{non-Arch}(K<t,t^{-1}>,A)\overset{1:1}{\longleftrightarrow } U(A)
$$
\\
\\ 
When we talk about the canonical norm on $\Q_p^n$, the nature observation is to use the Galois theory. The usual norm on $\CC$ can be defined as:
$$|x|=\sqrt{x\bar{x}}=\left ( |det x| \right ) ^{\frac{1}{|\CC:\R|} }\ \ \forall x\in \CC$$
We view $x$ as a $\R$-Linear transform on $\CC$, the determinant of $z$ is multipliable and we can check the triangle inequality of $\sqrt{z\bar{z}}$.
Similarly, suppose K is a finite Galois field extension of $\Q_p$ with degree $n$. There exists a canonical multipliable norm:
$$|x|=(\prod_{\sigma \in \mathrm{Gal}(K|\mathbb{Q}_p) }\sigma(x) )^{\frac{1}{|K:\mathbb{Q}_p|}}=\left ( |det x|\right ) ^{\frac{1}{|K:\Q_p|} }\ \ \forall x \in K$$
This multipliable norm is non-Archimedean and unique. In field extension theory of $\Q_p$, there exists a kind of isotropy Galois extension of $\Q_p$ named unramified extension. Such extension is unique. There exists a $\Q_p$-base(not unique) of $K$ named $\left \{  e_1,e_2...e_n \right \}$ such that the norm can be written as a supremum norm:
$$x=\sum_{i=1}^{n} x_ie_i,x_i \in \Q_p,\forall x \in K$$
$$|x|=\sup_{1\le i \le n}   \left (|x_i|\right )$$
There is an another way to show the supremum norm on $\Q_p^n$ is canonical. Let $GL_n(\R)$ be the group of linear invertible transform on $\R^n$, $GL_n(\Q_p)$ be the group of linear invertible transform on $\Q_p^n$. Then we have:
\begin{thm}
$O_n(\R)$ is the maximal compact subgroup of $GL_n(\R)$ .$GL_n(\Z_p)$ is the maximal compact subgroup of $GL_n(\Q_p)$. The maximal compact subgroup is unique up to conjugate.
\end{thm}
We can prove that $O_n(\R)$ is the isometric group of $\R^n$ (equiped with the usual quadric norm). $GL_n(\Z_p)$ is the isometric group of $\Q_p^n$ (equiped with the Non-Archimedean supremum norm). In this sense, the norm above is canonical.\\ \\
Let $Norm(\Q_p^n)$ be the set of Non-Archimedean norms such that they take values in $p^{\Z}$. Let $Norm(\R^n)$ be the set of Archimedean norms such that they can induce a inner product.
\begin{thm}
There is a one to one correspendence: \ 
$$Norm(\Q_p^n) \overset{1:1}{\longleftrightarrow } GL_n(\Q_p)/GL_n(\Z_p) \overset{1:1}{\longleftrightarrow } Lat(\Q_p^n)$$
$$Norm(\R^n) \overset{1:1}{\longleftrightarrow } GL_n(\R^n)/O_n(\R) \overset{1:1}{\longleftrightarrow } Sym_{+}(\R^n)$$
\end{thm}
$Lat(\Q_p^n)$ is the set of lattice in $\Q_p^n$. A lattice in $\Q_p^n$ is a $\Z_p$- submodule $L$ such that:$\Q_p \otimes_{\Z_p}L=\Q_p^n$ and $L$ is compact. (So it is isomorphic to $\Z_p^n$)
$Sym_{+}(\R^n)$ is the set of positive definite symmetric matrix of $GL_n(\R^n)$.

\bibliographystyle{plain}
\bibliography{Spectral theory of p-adic Hermite operator.bbl}

\end{document}